\documentclass[11pt,rqno,oneside]{extarticle}

\usepackage{mathptmx} 

\usepackage{titlesec}
\titleformat*{\section}{\Large\bfseries} 
\titleformat*{\subsection}{\large\bfseries} 

\usepackage{cite}
\usepackage[small]{caption}
\usepackage{amsmath}
\usepackage{amssymb}
\usepackage{amsfonts}
\usepackage{setspace}
\usepackage{latexsym}
\usepackage{mathrsfs}
\usepackage{epsfig}
\usepackage{amsthm}
\usepackage{yfonts}
\usepackage{graphicx}
\usepackage{color}

\usepackage{MnSymbol}

\usepackage{enumitem}
\setitemize{wide}

\numberwithin{equation}{section}

\newcommand{\n}{\noindent}
\newcommand{\be}{\begin{equation}}
\newcommand{\ee}{\end{equation}}
\newcommand{\ben}{\begin{displaymath}}

\newcommand{\een}{\end{displaymath}}

\newcommand{\vs}{\vspace{0.2cm}}

\usepackage{fancyhdr}
\AtBeginDocument{\thispagestyle{plain}}
\fancypagestyle{plain}{
\fancyhead{}
\fancyfoot{}
\fancyhead[LE,RO]{}
\fancyhead[LO,RE]{}
\fancyfoot[R]{\thepage}

\headsep = 20pt
}

\newcommand{\body}{\mathcal{B}}
\newcommand{\mass}{M}
\newcommand{\torus}{{\rm T}}
\newcommand{\dist}{{\rm dist}}
\newcommand{\bmass}{\mathcal{M}}
\newcommand{\radius}{\mathcal{R}}
\newcommand{\ssurf}{{\rm S}}
\newcommand{\hyper}{\mathcal{S}}

\usepackage{amsbsy}

\usepackage[margin=2cm]{caption} 

\newtheorem{Theorem}{Theorem}[section]

\newtheorem{Proposition}{Proposition}[section]
\newtheorem{Lemma}{Lemma}[section]
\newtheorem{Corollary}{Corollary}[section]
\newtheorem{Aux-Lemma}{Aux-Lemma}[section]

\begin{document}

\n \begin{minipage}[l]{11cm}
\n {\huge On the shape of bodies in General 

\vs
\n Relativistic regimes}

\vspace{.1cm}

\vspace{.3cm}
\n {\sc\large Martin Reiris}\\
\n {email: martin@aei.mpg.de}\\

\n \textsc{Max Planck Institute f\"ur Gravitationsphysik \\ Golm - Germany}\\

\begin{spacing}{.9}{\small The analysis of axisymmetric spacetimes, dynamical or stationary, is usually made in the reduced space. 
We prove here a {\it stability property} of the quotient space and use it together with minimal surface techniques to constraint the shape of General Relativistic bodies in terms of their energy and rotation. These constraints are different in nature to the mechanical limitations that a particular material body can have and which can forbid, for instance, rotation faster than a certain rate, (after which the body falls apart). The relations we are describing instead are fundamental and hold for all bodies, albeit they are useful only in General Relativistic regimes. For Neutron stars they are close to be optimal, and, although precise models for these stars display tighter constraints, our results are significative in that they do not depend on the equation of state.}
\end{spacing}

\vspace{.3cm}
\hspace{.1cm} 

\end{minipage}

\section{Introduction.} The purpose of this article is twofold. The first goal, technical in nature, is to prove a fundamental {\it stability property} of the reduced space of axisymmetric spacetimes. Material bodies, if present, are assumed to satisfy the dominant energy condition only. The second motivation is to use such stability, together with minimal surfaces techniques, to investigate a plethora of relations constraining  the energy, momentum or angular momentum of bodies in terms their geometry only. These relations are different in nature to the mechanical limitations that a particular material body can have and which can forbid, for instance, rotation faster than a certain rate, (after which the body falls apart). The relations we are describing are instead fundamental and hold for all bodies, albeit their usefulness is only in General Relativistic regimes.     

Ever since the dispute between Newton and Cassini on the shape of the earth, decided victoriously on Newton's side (see \cite{MR0213075} and ref. therein), the interest on the relation between mass, rotation and the shape of bodies, never ceased. Great attention was paid in particular to the analysis of self gravitating {\it figures of equilibrium}, with the obvious importance for astronomy. A large list of conspicuous physicist and mathematicians contributed  to the study. Notoriously, Newton himself, MacLaurin, Dirichlet, Riemann, Poincar\'e and Chandrasekhar, among many others (see \cite{MR0213075} for an historical account). But all this is the realm of Newtonian mechanics. With the advent of numerical simulations, the classical lines of investigation continued vigorously into General Relativistic regimes. Mainel et al. \cite{MR2441850}, for instance, did an exhaustive study of {\it relativistic figures of equilibrium} for fluids, including systems where matter and black holes are in mutual balance. The literature on the subject is indeed very big. The survey \cite{Stergioulas:2003yp} by Stergioulas, contains a complete discussion on the subject.     

What it seems to be not properly accounted in the literature, is that the shape of these relativistic figures of equilibrium are tightly linked to their energy or their angular momentum through definite and sometimes simple rules. The relations we are referring to do not depend on the equation of state of matter and are the expression of the structure of the General theory of Relativity, rather than of the particular material model. The present work intends to investigate these relations in a more systematic way.

Many criteria on black hole formation existent in the literature can be used in the form of geometry-energy relations for bodies in static equilibrium. For instance in \cite{Bizon:1989xm}, Bizon, Malec and O'Murchadha stated an outstanding condition (an inequality), which, when it holds, proves the presence of an apparent horizon. The result applies to dynamical but spherically symmetric systems as seen over maximal and asymptotically flat slices. As a system consisting of a body resting in static equilibrium is not a black hole, then the opposite inequality must hold for it. The relations we seek in this article are relatives of this last one, but apply in dynamic situations as well and are not restricted necessarily to spherical symmetry. 

In a seminal work \cite{BHFCM}, Schoen and Yau argued somehow in the opposite direction. First they defined a geometric radius $\radius_{\rm SY}$ for bodies. This radius depends innately on the time-slice, or, loosely speaking, on the ``instant at which the body is observed''. It was then shown that if the slice is maximal, then we have
\be
\radius_{\rm SY}^{2}(\body)\leq \frac{\pi}{6\rho_{0}},
\ee
provided that $\rho-|j|\geq \rho_{0}>0$ holds on $\body$. Finally, this inequality is used to prove that, if on a certain time-slice (necessarily non-maximal) we have $\radius^{2}_{\rm SY}(\body)>\pi/6\rho_{0}$ then an apparent horizon (past or future) is present. 

The Schoen-Yau argument can be summarised as a procedure. It starts by proving a geometric inequality for bodies on maximal slices and then promote it, (by a suitable use of the Jang equation), to a black hole formation criteria (on non-maximal slices). This procedure is by now standard, (see \cite{Khuri:2009dt} and references therein). The results of this paper deal with bodies on maximal slices and could (in principle) be used to obtain criteria for black hole formation through the mentioned procedure. This is particularly relevant considering that part of this article focuses on angular momentum and that, until now, no black hole formation criteria exists based on it.   

The work of Schoen and Yau was resumed by O'Murchadha in \cite{PhysRevLett.57.2466}. There it was emphasised for the first time that the Schoen and Yau radius can be useful to constraint the size of stars in terms of their energy density. Later, Klenk \cite{Klenk} used similar ideas to study size constraints for rotating stars, although he assumed the stability of the equatorial disc which is not justified at high densities (see the {\it Discussion} in \cite{Klenk}). Recently, Dain \cite{Dain:2013gma} argued on the existence of a meaningful definition of the radius $\radius(\body)$ of bodies for which the universal relation,
\be\label{DAING}
|J(\body)|\lesssim \frac{c^{3}}{G}\radius^{2}(\body)
\ee
holds, at least on maximal slices. Some of the ideas of the present article are based in the work \cite{Dain:2013gma}. We will also validate part of it. 

The Schoen and Yau radius, or the O'Murchadha radius, were criticised fundamentally on the grounds of meaning and computability.
What do these radii measure in complicated rotating bodies, and how are they related to simple magnitudes of a body like diameter, surface-area or volume? The answer to this question depends very much on the behaviour of stable minimal discs embedded in the region enclosed by the body, and this is in general very difficult to determine and to relate. 

Nevertheless, in the context of axisymmetry, (the one on which we are going to concentrate), there is a turnaround allowing to use minimal surface techniques in a more sophisticated way and to extract relevant information of familiar geometric magnitudes. The crucial fact to realise is that the quotient space of maximal slices satisfies a stability property identical to that of stable minimal surfaces on the ambient three-space. Techniques of minimal surfaces can then be applied and geometric information of the ambient space can then be obtained. 

Let us explain briefly this fundamental fact. It will be called the {\it stability property of the reduced space} an is discussed in full extent in Section \ref{SOQ-S}.  

This article concerns the geometry of bodies $\body$ when they are observed over maximal, axisymmetric and asymptotically flat slices $\mathcal{S}$. The bodies are not assumed to be in stationary equilibrium. For simplicity we assume that $\mathcal{S}$ is diffeomorphic to $\mathbb{R}^{3}$. Consider a stable (non-necessarily axisymmetric) compact surface $\ssurf$ embedded in $\mathcal{S}$. The surface can have boundary or not. In this context,  
for all real functions $\alpha$ of compact support in the interior of $\ssurf$ we have
\be\label{SPQS}
\int_{\ssurf} \big[ |\nabla \alpha|^{2}+\kappa \alpha^{2}\big]\, dA\geq \frac{1}{2}\int_{\ssurf} \big[ |K|^{2}+16\pi\rho\big]\, \alpha^{2}\, dA
\ee
where $\kappa$ is the Gaussian curvature of the induced two-metric on $\ssurf$. 

Now, quotient the three-manifold $\hyper$ by the action of the rotational Killing field $\xi$, denote the quotient surface by $\Sigma$ and the induced (quotient)-metric by $\gamma$. What the {\it stability property of the quotient space} says is that exactly (\ref{SPQS}) holds too on $(\Sigma,\gamma)$. That is, to a formal extent the surface $(\Sigma,\gamma)$ can be considered as a stable minimal surface, although $(\Sigma,\gamma)$ is not necessarily isometric to a minimal plane embedded in $\hyper$. What is important about this property is that this ``imaginary'' $(\Sigma,\gamma)$ is easily related to the geometry of $(\hyper, g)$.

We now explain which minimal surfaces  techniques will be used to exploit the stability property of the quotient space. The most basic is the radius estimate used in \cite{BHFCM} and which is originally due to Fischer-Colbrie \cite{MR808112}. A proof of it can be found in the survey \cite{MR2483369} by Meeks, P\'erez and Ros, (Thm 2.8). In our context, the drawback of this estimate is that it is useful only when $\rho\geq \rho_{0}>0$. To get general bounds we need to use (\ref{SPQS}) with convenient trial functions. {\it Radial functions}, which by definition depend only on the geodesic distance to a point or to a closed curve, give optimal results. These type of functions have been used in the past in a long list of works\footnote{Not all of them with accurate list of references.}. This includes, Pogorelov \cite{MR630142}, Lawson and Gromov \cite{MR720933}, Kawai \cite{MR945852} and Colding and Minicozzi \cite{MR1877004}. We will use a more general result, (but a relative to the ones before), due to Castillon \cite{MR2225628} and which is based upon the work of Shiohama and Tanaka (see for instance \cite{MR2028047}). We refer again to the survey \cite{MR2483369} (Thm 2.9) for a related presentation. A simple generalisation in arXiv:1002.3274 will also be necessary.  

The next section discusses the applications of the stability property to the study of bodies that we will prove in this article. 

\subsection{Main applications.}\label{MAINAPPS}
We start with a simple application to spherically symmetric spacetimes, (which are of course axisymmetric for many possible choices of the rotational axis). We consider then a spherically symmetric body $\mathcal{B}$, a star let's say, as seen on an asymptotically flat maximal hypersurface. What we observe here is that, if the density $\rho$ is everywhere greater than $\rho_{0}$, then the area of any spherical section of the body, (i.e. the spheres of points equidistant to the centre of $\mathcal{B}$), is always less or equal than $8\pi^{2}/3\rho_{0}$. The surface of the body, in particular,  satisfies this bound.
\begin{Theorem}\label{THSP} Let $\mathcal{B}$ be a spherical body as seen on an asymptotically flat maximal slice. If the energy-density $\rho$ is greater or equal than $\rho_{0}>0$, then the area $A$ of any constant-radius sphere of $\mathcal{B}$, (in particular the area of its surface), satisfies
\be\label{ADB}
A\leq \frac{8\pi^{2}}{3\rho_{0}}.
\ee
\end{Theorem}

Thus, the bigger the (uniform) density, the smaller the area of the surface of the body. It is worth noting that this result is valid in any regime, dynamical or static, and that it makes no assumption on the equation of state of the matter which is usually very speculative at the densities where the bound (\ref{ADB}) is relevant. 

Imagine for instance a star whose density is everywhere no less than $12\times 10^{17}{\rm kg/m}^{3}$. 
This density is the one that a Neutron star can have at its Inner Core and which is comparable to four times the density of an atomic nucleus. Then, according to Theorem \ref{THSP}, the areal-radius of such a star cannot be bigger than $50$ kilometres. This is just six times the radius of the Inner Core of a Neutron star (believed to be around 8 km). 
The density and radius of Neutron stars are therefore at the frontier of what can exist, independently of what the matter in the universe is made of. The close bound of 50 km and the insensitivity to the matter model, suggest that the Theorem \ref{THSP} could be relevant in situations well beyond those of only academic interest. 

The method of proof of Theorem \ref{THSP} allows us also to prove the inequality
\be\label{AOMI}
A\leq 16\pi \mathcal{R}^{2}_{\rm O'Mur}\leq \frac{8\pi^{2}}{3\rho_{0}}
\ee
where $A$ is the area of any constant-radius sphere of the body, (in particular the area of the surface), and $\mathcal{R}_{\rm O'Mur}$ is the O'Murchadha radius, (see Section \ref{PISS}). This says that at least in some situations, in particular in spherical symmetry, the O'Murchadha radius is a good measure of size. 

\vs
Theorem \ref{THSP} also gives a curious criteria for black-hole formation when the density $\rho$ is compared to the ADM-mass $M$. The following corollary to Theorem \ref{THSP} explains this fact.
\begin{Corollary}\label{COROLR}
(Under the hyp. of Thm. \ref{THSP}). If the density $\rho$ of the body is everywhere greater than $\pi/6M^{2}$, that is, if 
\be
\rho> \frac{\pi}{6M^{2}},
\ee
then the body lies entirely inside a black-hole, (past or future), and is not in static equilibrium.
\end{Corollary}
In particular, a star in equilibrium with $\rho\geq \rho_{0}$ must have $M^{2}\leq \pi/6\rho_{0}$. As above, one can estimate this bound for an ideal Neutron star with $\rho_{0}\sim 12\times 10^{17}$kg/m$^{3}$. We get for this that the ADM-mass $M$ can be at most 16 solar masses. This is another interesting result as it does not depend on the equation of state of the matter.

Corollary \ref{COROLR} is proved as follows. If $\rho>(\pi/6M^{2})$, then $A(\partial \body)<16\pi M^{2}$ by Theorem \ref{THSP}. But outside $\body$ the spacetime is the Schwarzschild spacetime, hence, if $\partial \body$ is not inside a past or future black hole we should have $A\geq 16\pi M^{2}$ which is not the case. 

\vs
We move now to study systems that are in rotation and which are consequently non-spherically symmetric. The systems will be just axially symmetric. We analyse two different situations separately depending on whether the rotating body intersects the axis or not. We analyse first the case when the object doesn't intersect the axis and use the results to study the other situation. Both contexts are equally important.

When rotation is present the shape of bodies is set by the competition between the angular momentum, which tries to expand, and the density which tries to contract. As a result the shape of rotating bodies must be studied with three magnitudes of length. 
\begin{figure}[h]
\centering
\includegraphics[width=6.5cm,height=4.5cm]{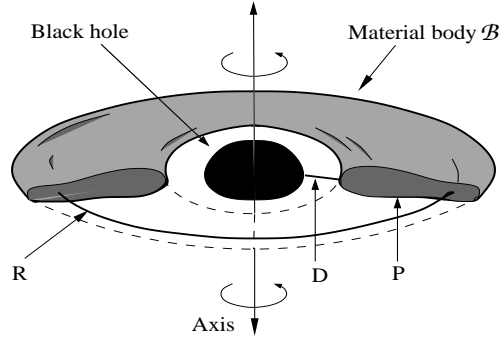}
\caption{A black-hole surrounded by a (connected) material body is schematised. The length $R$ of the greatest axisymmetric circle (which doesn't have to be the ``outermost'' one in the figure), the sectional perimeter $P$ and the length to the axis $D$ are also indicated. The length minimising geodesic from the body to the axis, whose length is $D$, enters indeed the black hole (apparent horizon).}
\label{Figure1}
\end{figure}
As said, we start with a very general geometric estimation of the angular momentum of axisymmetric (connected) bodies that do not intersect the axis. An example of one such body, actually surrounding a black-hole, is given in Figure \ref{Figure1}. Stationary systems with the same silhouette have been simulated and analysed in \cite{MR2441850}. Also in Figure \ref{Figure1} are indicated the three main magnitudes of length, $R,P$ and $D$. 

We study the geometry of the body $\body$ as seen inside a maximal, axisymmetric and asymptotically flat slice $\mathcal{S}$. For simplicity we assume that $\mathcal{S}$ is diffeomorphic to $\mathbb{R}^{3}$. The quotient of $\mathcal{S}$ by the rotational Killing field is denoted by $\Sigma$ and is (assumed to be) diffeomorphic to the half-plane. The projection from $\mathcal{S}$ into $\Sigma$ is denoted by $\Pi:\mathcal{S}\rightarrow \Sigma$.

In this context the definitions of $R,D$ and $P$ are

\begin{enumerate}
\item $R$ is the length of the greatest axisymmetric orbit (circle) in $\body$, 

\item $D$ is the distance from $\body$ to the the axis, and, 

\item $P$ is the transversal perimeter which is defined as follows. If $\body$ is connected, then $\partial(\Pi(\body))$ consists of a finite set of closed curves, one of which (and only one of which) encloses all the other. Denote such closed curve by $\partial^{e}(\Pi(\body))$. Then, the length of the smallest closed curve in $\partial \body$ and projecting into $\partial^{e}(\Pi(\body))$ is the transversal perimeter $P$. 
\end{enumerate}
The Kommar angular momentum $J$ of $\body$ is 
\be
J(\body)=\int_{\body} {\bf T}(n,\xi)\, dV
\ee 
where ${\bf T}$ is the stress-energy tensor, $n$ is a unit timelike normal to $\Sigma$, and $\xi$ is the rotational Killing field. Throughout this article we will use this notion of angular momentum.

\vs
The following is our most general result for bodies not intersecting the axis.

\begin{Theorem}\label{AST1} Let $\body$ be an axisymmetric body as seen on an asymptotically flat maximal slice. If $\body$ does not intersect the axis and is connected, then the angular momentum $|J|$ of $\body$ satisfies, 
\be\label{EEE}
8\pi\, |J|\leq \bigg(1+\frac{P}{\pi D}\bigg)\, R^{2}.
\ee
\end{Theorem}
This says that, fixed the distance $D$ to the axis, then an increase in $|J|$ implies an increase in $R$ or $P$. We stress that this is a completely general statement that makes no assumption on the type of matter, the profile of the density $\rho$ or the momentum-current $j$. In particular it doesn't make use of any lower bound on $\rho$ (as Theorem \ref{THSP} did). What is extraordinary about (\ref{EEE})
is that the bound on the angular momentum is strictly geometrical.

We would like to mention that a simple modification of the proof of Theorem \ref{AST1} gives also the following geometric bound on the {\it proper mass} $\bmass$, (sometimes called {\it baryonic mass} too),
\be\label{EEM}
\bmass=\int_{\body} \rho dV\leq \bigg(1+\frac{P}{\pi D}\bigg)\, \frac{R}{8}
\ee
When the gravitational binding energy is negative, the proper mass is greater or equal than the ADM-mass $M$ and (\ref{EEM}) gives also a geometric bound for $M$. Thus occurs for instance when the body is in static equilibrium, (as can be seen by integrating the (maximal) Lapse equation and using that the Lapse is less or equal than one). 

\vs
In some important situations the dependence on the transversal perimeter $P$ and the connectedness of $\body$ can be eliminated altogether. We explain this important point below.

Given an axisymmetric circle $C$ at a distance $d=\dist_{g}(C,{\rm Axis})$ from the axis, and a number $a<d$, then the set of point at a distance less than $a$ from $C$ will be denoted by $\torus(C,a)$, that is
\be
\torus(C,a)=\big\{p: {\rm dist}_{g}(p,C)<a\big\}
\ee
When the metric $g$ is almost flat then the $\torus(C,a)$ are most likely to be solid tori, that is, topologically the product of a two-disc and a circle, ($\mathbb{D}^{2}\times \mathbb{S}^{1}$). In other instances this does not have to be the case.

The next results investigate the angular momentum of a body $\body$ when one knows that it lies inside a particular region $\torus(C,a)$. Of course, this is no more than assuming some a priori global proportions on the main dimensions on the body. 
What is interesting is that this allows a complete geometric estimation of the angular momentum. An example of such a situation is represented in Figure \ref{Figure2}.
\begin{Theorem}\label{ASPP} Let $\body$ be an axisymmetric body as seen on an asymptotically flat maximal slice. If $\body\subset \torus(C,a)$, then
\be
8\pi\, |J|\leq \bigg(\frac{1}{1-a/d}\bigg)^{2} R^{2}
\ee
where $d=\dist_{g}(C,{\rm Axis})$.
\end{Theorem}

Contrary to Theorem \ref{AST1} the body $\body$ in Theorem \ref{ASP} doesn't have to be connected. This is definitely an advantage.

\begin{figure}[h]
\centering
\includegraphics[width=6cm,height=5cm]{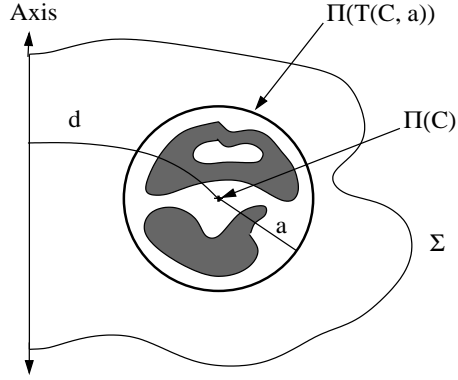}
\caption{Illustration of the main elements as in Theorem \ref{ASPP}. The are two shaded regions representing the projection of the body $\body$ to $\Sigma$, i.e. $\Pi(\body)$. In this case $\body$ has two components one of which is a solid torus.}
\label{Figure2}
\end{figure}

A particular case of Theorem \ref{ASPP} is the following Corollary \ref{ASP}.
\begin{Corollary}\label{ASP} Let $\body$ be an axisymmetric body as seen on an asymptotically flat maximal slice. If $\body\subset \torus(C,a)$ with $2a<d=\dist_{g}(C,\rm{Axis})$, then 
\be
8\pi\, |J|\leq 4R^{2}
\ee
\end{Corollary}
As in Theorem \ref{ASP}  the body $\body$ in Corollary \ref{ASPP} doesn't have to be connected.

\vs
For bodies intersecting the axis the estimations given before cannot be immediately used, ($D=0$ in this case), although they can be used to obtain bounds on the angular momentum carried by toroidal subregions of the body. However, and as we will see, an ingenious application is still possible from which information on the whole body can be obtained. To this end we ask the following question: Can a body have arbitrarily large angular momentum if it is known that its metric tensor is constrained? 

To clarify the extent of this question, let us start by analysing a very simple situation in Newtonian mechanics\footnote{In this argumentation we take some ideas from \cite{Dain:2013gma}.}. Imagine a rotating body $\body$ whose geometry is known to be that of a perfect solid sphere in Euclidean space. Suppose too that the area of its surface is $A$. If the mass density $\rho$ is constant then a straightforward computation gives
\be\label{AJN}
|J|=\bigg[\frac{\Omega M}{10\pi}\bigg]\, A,
\ee
where $M$ is the total mass and $\Omega$ is the angular velocity. As in Newtonian mechanics $M$ and $\Omega$ are unconstrained, we see from this example that $|J|$ and $A$ are in general unrelated. Non surprisingly, $|J|$ is not limited by the constraint on the geometry. However, the situation changes if we borrow from General Relativity some heuristic limitations on $M$ and $\Omega$. Thus, require that $\mathcal{R}\geq 2M$, (meaning that the system is not a black-hole; here $4\pi \mathcal{R}^{2}=A$), and require that $\mathcal{R}\Omega\leq c$, (meaning that no point in the body moves faster than the speed of light $c$). These assumptions transform (\ref{AJN}) into the suggestive inequality
\be\label{MBO}
|J|\leq \bigg[\frac{c^{3}}{20\pi G}\bigg]\, A,
\ee
where a bound for $|J|$ in terms of $A$ is explicit. This heuristic inequality is obviously applicable also to any spherical region $\body'$ of $\body$ of radius $\mathcal{R}'<\mathcal{R}$
\be
\body'=\{(r,\theta,\varphi): r\leq \mathcal{R}'\}\subset \body.
\ee
Namely, we can expect,
\be\label{MBOK}
|J(\body')|\leq \frac{c^{3}}{20\pi G}A(\partial \body').
\ee
This simple observation will be relevant for a later comparison.

Thus, on the base of (\ref{MBO}) we find it justified to ask if a similar bound can be proved within General Relativity when one knows beforehand how the geometry of the body is constrained. In this example we assumed that the body was constrained to be a perfect sphere in rotation.

In General Relativity it is not possible to assume that the metric of a body is flat because the energy constraint would imply $\rho=0$ and there would be no object after all. So let us assume that we know that the metric of the body is constrained in the following way:
\be\label{MB}
\Lambda^{-2}g_{F}\leq g\leq \Lambda^{2}g_{F}
\ee
where $g_{F}=dr^{2}+r^{2}(d\theta^{2}+\sin^{2}\theta d\varphi^{2})$ is the flat Euclidean metric of a coordinate patch $\{(r,\theta,\varphi)$, $r\leq \mathcal{R}$\} covering the body, that is
\be
\body=\{(r,\theta,\varphi): r\leq \mathcal{R}\}.
\ee
We can assume, in principle, any $\Lambda>1$ but for the sake of concreteness let us set $\Lambda=1.1=11/10$. With this choice, the body $\body$ is metrically constrained to be close to a perfect solid sphere, as we were arguing until now. Does the assumption (\ref{MB}) imply a bound in $|J|$ as in (\ref{MBO})?

Remarkably, under (\ref{MB}) only, we can prove that for any $\mathcal{R}'\leq 2\mathcal{R}/3$, the angular momentum $|J(\body')|$ carried by the (topologically) spherical region $\body'$ of $\body$, 
\be
\body'=\{(r,\theta,\varphi): r\leq \mathcal{R}'\}\ \subset \body,
\ee
is bounded by the surface-area $A(\partial \body')$ of the same region $\body'$ as 
\be\label{GEOMB}
|J(\body')|\leq \bigg[\frac{c^{3}}{20\pi G}\bigg]\, A(\partial \body'),
\ee
which is exactly what we were expecting from (\ref{MBOK}). What is striking here is that the constraint (\ref{MB}), which is just on the metric, implies (\ref{GEOMB}) without any assumption on the energy density $\rho$ or the stress-energy tensor, (we use just $0\leq |j|\leq \rho$). Observe that no reference whatsoever is made in this statement about the exterior of the body. This is all remarkable. The price paid however, is that the bound is for the angular momentum carried by the central parts of the body but not by the whole body itself. Finally note that we obtain Dain's guess (\ref{DAING}) if the area in (\ref{GEOMB}) is replaced by the areal radius.  

The bound (\ref{GEOMB}) can very well be named, {\it core estimates}. They could be relevant in the analysis of millisecond pulsars, as for them the right and left hand sides of (\ref{GEOMB}) are of the same order of magnitude.  

Let us see how the argument to prove (\ref{GEOMB}) works. The Figure \ref{Figure3} shows the projection of the body $\body$ and the subregion $\body'$ into the coordinate patch $(r,\theta)$, (we eliminated $\varphi$ when passing to the quotient). 
As can be seen also in this figure we have ideally divided $\body'$ in regions $\body'_{i}$, $i=1,2,3,4,\ldots \infty$.
To bound $|J(\body')|$ we will first use (\ref{AST1}) to estimate the angular momentum carried by each one of the regions $\body'_{i}$, and then add the contributions up. Thus, we will use
\be\label{SUM}
|J(\body')|\leq \sum_{i=1}^{i=\infty} |J(\body_{i}')|
\ee
\begin{figure}[h]
\centering
\includegraphics[width=7cm,height=6cm]{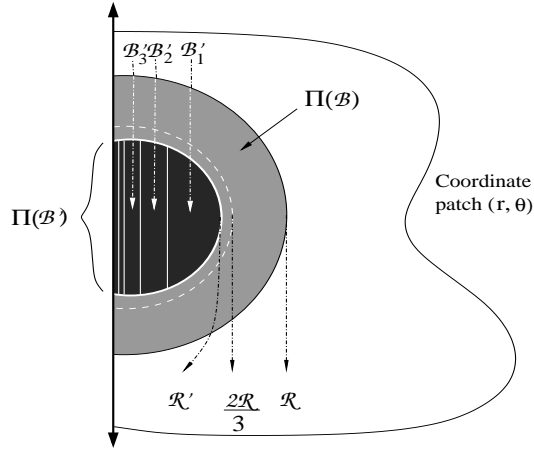}
\caption{The figure illustrates the main construction behind the proof of the inequality (\ref{GEOMB}). The projection of the body $\body$ and the interior region $\body'$ into the coordinate path $(r,\theta)$ is given. The fact that they are half discs is due only to the choice of coordinates and represents nothing geometrical.}
\label{Figure3}
\end{figure}

Each region $\Pi(\body'_{i})$ is defined as 
\be
\Pi(\body'_{i})=\Pi(\body')\cap \big\{(r,\theta): \frac{\mathcal{R}'}{2^{i}}\leq r\sin\theta \leq \frac{\mathcal{R}'}{2^{i-1}}\big\}
\ee
(note that the set on the right is just a vertical stripe of width $\mathcal{R}'/2^{i}$). Then, for each $\Pi(\body'_{i})$ we have
\begin{align}
& R_{i}\leq \frac{\Lambda \mathcal{R}'}{2^{i-1}},\\
& P_{i}\leq \Lambda (2+\pi)\mathcal{R}',\\
& D_{i}\geq \frac{\mathcal{R}'}{\Lambda 2^{i}}
\end{align}
Plugging these inequalities in (\ref{AST1}) we deduce
\be
8\pi |J(\body'_{i})|\leq \bigg[1+\frac{(2+\pi)\Lambda^{2}}{2^{i}\pi}\bigg]\frac{\Lambda^{2}\mathcal{R}'^{2}}{2^{2(i-1)}}
\ee  
and using this in (\ref{SUM}) we obtain
\be
8\pi |J(\body')|\leq \bigg[\frac{4}{3}+\frac{4(2+\pi)\Lambda^{2}}{7\pi}\bigg]\Lambda^{2}\mathcal{R}'^{2}
\ee
But as $A(\partial \body')\geq 4\pi\mathcal{R}'^{2}/\Lambda^{2}$ and $\Lambda=11/10$ we get, (after a computation),
\be
|J(\body')|\leq \frac{1}{20\pi}A(\partial \body')
\ee
as wished. 

\vs
We would like to mention that a the same reasoning, but using (\ref{EEM}) instead of (\ref{EEE}), leads to the bound
\be
\frac{\pi}{5}\bmass^{2}(\body')\leq A(\partial \body')
\ee
between the proper mass contained in $\body'$ and the area of its surface. If the binding gravitational energy is negative then $\bmass\geq \mass$ and we get $A\geq (\pi/5)M^{2}$. What these inequalities say is that the geometry of a body constraints also the amount of mass (energy) that it can carry. 

\vs
The argument above, leading to (\ref{GEOMB}), was made for $\Lambda=11/10$ and $\mathcal{R}'\leq 2\mathcal{R}/3$, but nothing about these assumptions was used in an essential way. The following Theorem generalises the bound (\ref{GEOMB}) to any $\Lambda>1$ and any $\mathcal{R}'<\mathcal{R}$ but the result is not (cannot be) as nice as (\ref{GEOMB}) simply because $\Lambda >> 1$ implies that the geometry of $\body$ can deviate significantly from that of a perfect sphere in Euclidean space. The proof is based on a simple adaptation of the argument above and is left to the readers.

\begin{Theorem}\label{THHU} Let $\body$ be a topologically spherical body, as seen on a maximal, axisymmetric and asymptotically flat slice. Suppose that $\body$ is described as $\body=\big\{(r,\theta,\varphi): r\leq \mathcal{R}\big\}$ for some coordinates $(r,\theta,\varphi)$ and that over $\body$ we have
\be
\Lambda^{-2}g_{F}\leq g\leq \Lambda^{2}\, g_{F},
\ee
where $g_{F}=dr^{2}+r^{2}(d\theta^{2}+\sin^{2}\theta d\varphi^{2})$ is the Euclidean metric, and $\Lambda\geq 1$. Then, for every $\mathcal{R}'<\mathcal{R}$ there is $\mathcal{J}(\mathcal{R}',\Lambda)<\infty$ such that the angular momentum of the internal region $\body_{\mathcal{R}'}=\{(r,\theta,\varphi): r\leq \mathcal{R}'\}$ is bounded by $\mathcal{J}(\mathcal{R}',\Lambda)$. 
\end{Theorem}

\subsection{Related results also proved in this article.}

Another interesting avenue to measure the influence of angular momentum on bodies is through {\it enclosing isoperimetric spheres}. 
In many instances, the surface of an axisymmetric body $\body$, (that may or may not intersect the axis), is itself isoperimetric stable, meaning that it minimises area among volume-preserving variations. In other instances the body is enclosed by a nearby one such surface.  
Whatever the case, the angular momentum of the the body highly influences the shape of the isoperimetric surface. In this respect we are able to prove the following result.
\begin{Theorem}\label{ISESTT} Let $S$ be a stable isoperimetric, axisymmetric sphere enclosing a body $\body$, (and nothing else). Then, 
\be\label{ISEST}
|J|\leq {\rm c_{1}}\, R\sqrt{A}\leq {\rm c_{2}}\, R\, L
\ee
where ${\rm c_{1}}=\sqrt{6}/(8\pi^{3/2})$, ${\rm c_{2}}=\sqrt{6}/(4\pi)$, $|J|$ is the angular momentum of $\body$ and $A$, $R$ and $L$ are, respectively, the area of $S$, the length of the greatest axisymmetric orbit in $S$ and the distance from the North to the South pole of $S$.
\end{Theorem}
It is worth noting that black-hole apparent horizons, (i.e. stable MOTS), satisfy the inequalities
\be\label{ISESTBH}
|J|\leq {\rm c'_{1}}\, R\sqrt{A}\leq {\rm c'_{2}}\, R\, L
\ee 
with ${\rm c'_{1}}=1/(8\pi^{3/2})$ and ${\rm c'_{2}}=1/(32^{1/2}\pi^{1/2})$, which have the same form as (\ref{ISEST}) but with different constants. The inequalities (\ref{ISESTBH}) can be easily obtained by combining equations (16) and (17) in \cite{Reiris:2013jaa}.

\vs
When the surface of a body $\body$ is not isoperimetric stable or when there are no stable isoperimetric surfaces in its vicinity, then the Theorem \ref{ISESTT} doesn't say anything about the size of $\body$ itself. Yet we believe that a relation between $J$ and ``size'' should exist in general.

In General Relativity it is often the case that, to assess the validity of a statement involving angular momentum, (one that we do not know how to prove), a good alternative is start proving a similar statement for Einstein-Maxwell-(Matter) spacetimes but with $Q^{2}$ playing the role of $|J|$. With this in mind, let us consider a spherically symmetric charged material body whose exterior is also spherically symmetric and therefore modelled by the Reissner-Nordstr\"om spacetime, and ask whether the charge $|Q|$ of the body imposes any constraint on its size. By ``size'' we understand here no more than the area of the surface of the body $\body$. 
The following theorem answers this question\footnote{I would like to thank Sergio Dain for pointing out this problem to me.} affirmatively and gives further support to the belief that the angular momentum should impose strict constraints on the size of rotating bodies.
\begin{Theorem}\label{QATHM} Consider a spherically symmetric charged body $\body$ as seen on an asymptotically flat slice $\hyper$. The slice is not necessarily maximal and the stress energy tensor is assumed to satisfy the dominant energy condition. Then, the area $A$ of the surface of $\body$ satisfies
\be\label{MQA}
\frac{\pi Q^{4}}{M^{2}}\leq A,
\ee
where $Q$ is the total charge and $M$ is the ADM-mass. In particular if $M<|Q|$, that is, if the exterior of the body is a super-extreme Reissner-Nordstr\"om spacetime, then $\pi Q^{2}\leq A$.
\end{Theorem}


\section{Proofs.}
\subsection{The setup.}
We consider a maximal, axisymmetric and asymptotically flat, Cauchy hypersurface $\mathcal{S}$ of an axisymmetric spacetime $({\bf M};{\bf g};{\bf T})$. The stress-energy tensor of matter ${\bf T}$ is assume to satisfy the dominant energy condition. For simplicity we also assume that $\mathcal{S}$ is diffeomorphic to $\mathbb{R}^{3}$. The three-metric induced on $\hyper$ is denoted by $g$ and $n$ will be a unit, (timelike), normal to $\mathcal{S}$. The second fundamental form of $\mathcal{S}$ as a hypersurface of the spacetime (say, in the direction of $n$) is denoted by $K$. The axisymmetric Killing field is denoted by $\xi$ and its norm by $\lambda:=|\xi|_{g}$. 
The twist one-form $\boldsymbol{\omega}$ on ${\bf M}$ is defined by
\be
\boldsymbol{\omega}_{a}=\boldsymbol{\epsilon}_{abcd}\xi^{b}\boldsymbol{\nabla}^{c}\xi^{d}.
\ee
The energy density $\rho$ on $\mathcal{S}$ is $\rho={\bf T}(n,n)$ and the current $j$ associated to the linear momentum (i.e. a one-form in $\mathcal{S}$) is $j(v) = {\bf T}(n,v)$ for any $\nu\in {\mathcal T} \mathcal{S}$. We assume $|j|\leq \rho$, (pointwise).

Besides ${\bf M}$ and $\mathcal{S}$, there are two other relevant manifolds. 
\begin{enumerate}
\item The quotient of ${\bf M}$ by the action of $\xi$ is $\widetilde{M}$ and its quotient metric is $\tilde{g}$. The projector operator is $\widetilde{\Pi}_{a}^{\ b}={\bf g}_{a}^{\ b}-\xi_{a}\xi^{b}/\lambda^{2}$.

\item The quotient of $\mathcal{S}$ by the action of $\xi$ is $\Sigma$, (we include the axis in $\Sigma$), and its quotient two-metric is $\gamma$. As $n$ is perpendicular to $\xi$ then $n$ can be thought as a unit, (timelike), vector field over $\Sigma$. The second fundamental form of $\Sigma$, (in the direction of $n$), as a hypersurface of $(\widetilde{M},\tilde{g})$ is denoted by $\chi$. The Gaussian curvature of $\gamma$ is denoted by ($2\kappa=R_{\gamma}$) and $\nabla$ is its covariant derivative.

\end{enumerate}

\subsection{Proof of the stability of the quotient space.}\label{SOQ-S}
The {\it stability property of the quotient} will be deduced from the following proposition.

\begin{Proposition}\label{SQP} Let $\Delta$ and $\kappa$ be, respectively, the Laplacian and the Gaussian curvature associated to $\gamma$ on $\Sigma$. Then, we have
\be\label{SE}
-\Delta \lambda+\kappa\lambda=\bigg[8\pi \rho+\frac{|K|_{g}^{2}}{2}+\frac{\omega(n)^{2}}{4\lambda^{4}}\bigg]\, \lambda.
\ee
\end{Proposition}

\begin{proof}[\bf Proof.] The computation to obtain (\ref{SQP}) relies on the equations
\begin{align}
\label{EQ1} & \widetilde{\rm R}{\rm ic}_{ab}=\frac{\tilde{\nabla}_{a}\tilde{\nabla}_{b} \lambda}{\lambda} + \frac{\displaystyle \omega_{a}\omega_{b}-(\tilde{g}^{cd}\omega_{c}\omega_{d})\tilde{g}_{ab}}{{2\lambda^{4}}}+\widetilde{\Pi}_{a}^{\ c}\widetilde{\Pi}_{b}^{\ d}\boldsymbol{\rm\bf Ric}_{cd},\\
\label{EQ2} & \frac{\displaystyle \tilde{g}^{ab}\tilde{\nabla}_{a}\tilde{\nabla}_{b} \lambda}{\lambda}=-\frac{\displaystyle (\tilde{g}^{ab}\omega_{a}\, \omega_{b})}{2\lambda^{4}}-\frac{\boldsymbol{\rm\bf Ric}_{ab}\xi^{a}\xi^{b}}{\lambda^{2}},
\end{align}
and, 
\begin{align}
\label{EQ3} & \gamma^{ab} \chi_{ab}=\frac{n^{a}\tilde{\nabla}_{a}\lambda}{\lambda},\\
\label{EQ4} & \big|K|^{2}_{g}=(\chi_{ab}\chi_{cd}\gamma^{ac}\gamma^{bd})+\frac{(n^{a}\tilde{\nabla}_{a}\lambda)^{2}}{\lambda^{2}}+\frac{(\gamma^{ab}\omega_{a}\omega_{b})}{2\lambda^{4}}
\end{align}
The equations (\ref{EQ1}) and (\ref{EQ2}) are equivalent to equations (18.16) and (18.12) of \cite{MR2003646} respectively\footnote{The calculation in \cite{MR2003646} is for timelike Killing fields $\xi$, but the same formulae apply when, like in our case, $\xi$ is a rotational Killing field. Note too that $F$ in \cite{MR2003646} is $F=\lambda^{2}$.}. On the other hand, (\ref{EQ3}) and (\ref{EQ4}) are the equations (42) and (45) in \cite{Dain:2008xr} respectively. 

Contracting (\ref{EQ1}) with $\tilde{g}^{ab}$ and then using (\ref{EQ2}) gives
\be
\widetilde{\rm R}=\frac{2\tilde{g}^{ab}\tilde{\nabla}_{a}\tilde{\nabla}_{b} \lambda}{\lambda}-\frac{(\tilde{g}^{ab}\omega_{a}\omega_{b})}{2\lambda^{4}}+{\rm\bf R}
\ee
Also, contracting (\ref{EQ1}) with $n^{a}n^{b}$ gives
\be
n^{a}n^{b}\widetilde{\rm R}{\rm ic}_{ab}=\frac{\displaystyle n^{a}n^{b}\tilde{\nabla}_{a}\tilde{\nabla}_{b} \lambda}{\lambda} +\frac{(\gamma^{ab}\omega_{a}\omega_{b})}{2\lambda^{4}}+n^{a}n^{b}{\rm\bf Ric}_{ab}
\ee
From these two equations we obtain the combination
\begin{align}\label{2PS}
\nonumber \underbrace{2n^{a}n^{b} \widetilde{\rm R}{\rm ic}_{ab}+\widetilde{\rm R}\vphantom{\frac{}{}}}_{\rm (I)}\ =\ & \underbrace{2\big[\frac{\displaystyle n^{a}n^{b}\tilde{\nabla}_{a}\tilde{\nabla}_{b} \lambda}{\lambda} +\frac{\tilde{g}^{ab}\tilde{\nabla}_{a}\tilde{\nabla}_{b} \lambda}{\lambda}\big]}_{\rm (II)} +\underbrace{\frac{((n^{a}\omega_{a})^{2}+\gamma^{ab}\omega_{a}\omega_{b})}{2\lambda^{4}}\vphantom{\frac{}{}}}_{\rm (III)}\\ & +\underbrace{2n^{a}n^{b}{\rm\bf Ric}_{ab} +{\bf R}\vphantom{\big[}}_{\rm (IV)}.
\end{align}
To obtain (\ref{SE}) we manipulate the expressions (I), (II) and (IV) in the equation before. The term (IV) is equal to $16\pi \rho$ from the Einstein equations. On the other hand using
\be
n^{a}n^{b}\tilde{\nabla}_{a}\tilde{\nabla}_{b}\lambda+\tilde{g}^{ab}\tilde{\nabla}_{a}\tilde{\nabla}_{b}\lambda=\Delta \lambda+(\gamma^{ab}\chi_{ab})(n^{a}\tilde{\nabla}_{a}\lambda)=\Delta \lambda + (n^{a}\tilde{\nabla}_{a}\lambda)^{2}
\ee
we can transform the term (II) into $2(\Delta \lambda + (n^{a}\tilde{\nabla}_{a}\lambda)^{2})/\lambda$. Finally, for (I) we have 
\be
2n^{a}n^{b}\widetilde{\rm R}{\rm ic}_{ab}+\widetilde{\rm R}=2\kappa - (\chi_{ab}\chi_{cd}\gamma^{ac}\gamma^{bd})+(\gamma^{ab}\chi_{ab})^{2}
\ee
which using (\ref{EQ3}) and (\ref{EQ4}) can be transformed into
\be
2n^{a}n^{b}\widetilde{\rm R}{\rm ic}_{ab}+\widetilde{\rm R}=2\kappa - \big|K\big|^{2}_{g}+\frac{2(n^{a}\nabla_{a} \lambda)^{2}}{\lambda^{2}}+\frac{(\gamma^{ab}\omega_{a}\omega_{b})}{2\lambda^{4}}
\ee
Using these expressions for (I), (II) and (IV) in (\ref{2PS}), and making a pair of crucial cancellations, we obtain (\ref{SE}).
\end{proof}

Proposition \ref{SQP} gives us immediately the {\it stability property of the quotient}.

\begin{Lemma}[\bf The stability of the quotient]\label{SP} For any function $\alpha$ with compact support in the interior of $\Sigma$ we have
\be\label{SI}
\int_{\Sigma} \big[|\nabla \alpha|^{2} +\kappa \alpha^{2}\big]\, dA= \frac{1}{2} \int_{\Sigma} \bigg[ |K|_{g}^{2} + 16\pi \rho +\frac{\omega(n)^{2}}{2\lambda^{2}}\bigg]\, \alpha^{2}\, dA +\int_{\Sigma} \big|\alpha \nabla \ln \lambda - \nabla \alpha\big|^{2}\, dA.
\ee
\end{Lemma} 

\begin{proof}[\bf Proof] To obtain (\ref{SI}) divide (\ref{SE}) by $\lambda$, multiply by $\alpha^{2}$ and use the identity, (sometimes called Young's identity),
\be
-\int_{\Sigma} \frac{\Delta \lambda}{\lambda}\alpha^{2}\, dA=-\int_{\Sigma} \big|\alpha\nabla\ln \lambda-\nabla\alpha\big|^{2}\, dA+\int_{\Sigma} \big|\nabla\alpha\big|^{2}\, dA. \qedhere
\ee
\end{proof}

\subsection{Proofs in spherical symmetry.}\label{PISS}

Let us start by recalling the definition of the O'Murchadha radius $\radius_{\rm O'Mur}$ of a body $\body$. Let $D$ be a stable minimal disc embedded in $\body$ and let $\gamma$ be its induced metric. Define $rad(D)=\sup\{{\rm dist}_{\gamma}(p,\partial D): p\in D\}$. Then 
\be
\radius_{\rm O'Mur}(\body):=\sup\{rad(D):D \text{ stable minimal disc embedded in }\ \body\}.
\ee 
A fundamental estimate, due to Fischer-Colbrie \cite{MR808112} (see Thm. 2.8 in \cite{MR2483369}) and used by Schoen and Yau in \cite{BHFCM}, says that if $\rho\geq \rho_{0}$ then for every stable minimal disc $D$ we have $rad(D)\leq \sqrt{\pi/6\rho_{0}}$. Therefore, if $\rho\geq \rho_{0}$ then $\radius_{\rm O'Mur}(\body)\leq \sqrt{\pi/6\rho_{0}}$. We will use this estimate in the proof of Theorem \ref{THSP} below.

\begin{proof}[\bf Proof of Theorem \ref{THSP}.] For every constant-radius sphere of $\body$ let $s$ be its radius, (i.e. the distance to the centre of $\body$), let $A(s)$ be its area and set $r(s)=\sqrt{A(s)/4\pi}$, (i.e. $r$ is the areal-radius of the sphere). The radius of $\body$ is denoted by $s_{\body}$.   

We start by proving that for any $0\leq s_{1}<s_{2}\leq s_{\body}$ we have $s_{2}-s_{1}\geq r_{2}-r_{1}$, where $r_{i}=r(s_{i})$, $i=1,2$.
This follows from a nice observation due to Bizon, Malec and O'Murchadha (\!\cite{Bizon:1989xm}, pg. 965), stating that if we write the three-metric $g$ in the form
$g=\phi^{4}(d\tilde{r}^{2}+\tilde{r}^{2}d\Omega^{2})$ then the conformal factor $\phi=\phi(\tilde{r})$ is a monotonically-decreasing function of $\tilde{r}$. Indeed, using this observation we get
\be
s_{2}-s_{1}=\int_{\tilde{r}_{1}}^{\tilde{r}_{2}} \phi^{2}(\tau)d\tau\geq \phi^{2}(\tilde{r}_{2})(\tilde{r}_{2}-\tilde{r}_{1})
\geq \phi^{2}(\tilde{r}_{2})\tilde{r}_{2}-\phi^{2}(\tilde{r}_{1})\tilde{r}_{1}=r_{2}-r_{1}
\ee
as wished. 

Consider the disc formed by the intersection of $\body$ with the plane $\{(s,\theta,\varphi): \varphi=0\}$ which has two-metric 
$\gamma=ds^{2}+r^{2}(s)d\theta^{2}$. By {\it the stability property of the quotient}, half of the disc is stable. Namely, the domain $\mathcal{D}=\{(s,\theta): 0<s<s_{\body}\ \text{and}\ 0<\theta<\pi\}$ is stable. In particular, the distance from any of its points to its boundary is less or equal than $\mathcal{R}_{\rm O'Mur}$. 

As proved above, for any $\bar{s}\in [0,s_{\body}]$ we have $\bar{s}\geq r(\bar{s})>0$. Therefore, for any $\bar{s}\in (0,s_{\body}]$ we can consider the domain
\be 
\hat{\mathcal{D}}=\big\{(s,\theta): s\in (\bar{s}-\bar{r},\bar{s})\ \text{and}\ \theta\in (0,\pi)\big\}\subset \mathcal{D}
\ee
where we are making $\bar{r}:=r(\bar{s})$. Also, for any $s\in (\bar{s}-\bar{r},\bar{s})$ we have $\bar{s}-s\geq \bar{r}-r(s)$ and therefore $r(s)\geq \bar{r}-\bar{s}+s\geq 0$. Hence
\be
\gamma\geq ds^{2}+(\bar{r}+\bar{s}-s)^{2}d\theta^{2}
\ee
over $\hat{\mathcal{D}}$. Making $x=\bar{r}+\bar{s}-s$, we see from this that $\gamma\geq dx^{2}+x^{2}d\theta^{2}$ and that $\hat{\mathcal{D}}= \{(x,\theta): x\in (0,\bar{r})\ \text{and}\ \theta\in (0,\pi)\}$. Noting that $dx^{2}+x^{2}d\theta^{2}$ is just the Euclidean two-metric, we deduce that the distance from the point $(s,\theta)=(\bar{s}+\bar{r}/2,\pi/2)$, (that is, the point $(x,\theta)=(\bar{r}/2,\pi/2)$), to the boundary of $\hat{\mathcal{D}}$ must be greater or equal than $\bar{r}/2$. Hence,
\be
\frac{\bar{r}}{2}\leq \mathcal{R}_{\rm O'Mur}
\ee
But because $\bar{r}=r(\bar{s})$ and because $\bar{s}$ is any point in $(0,s_{\body}]$ we deduce that
\be
A\leq 16\pi\mathcal{R}^{2}_{\rm O'Mur}\leq \frac{8\pi^{2}}{3\rho_{0}}
\ee
as wished. 
\end{proof}

\subsection{Proofs for rotating systems.}\label{POR}
Below we will consider compact and connected regions $\Omega$ in ${\rm Int}(\Sigma)$, (${\rm Int}(\Sigma)$ is the interior of $\Sigma$), with smooth boundary $\partial \Omega$. For any such set we consider the set $\torus(\Omega)$ (simply $\torus$ from now on) in $\Sigma$ consisting of all the axisymmetric orbits in $\Sigma$ which project into $\Omega$. For instance if $\Omega$ is topologically a disc then $\torus$ is a solid torus around the axis of symmetry. 

Given $\Omega$ and $\torus=\torus(\Omega)$, the Kommar angular momentum $J(\torus)$ carried by $\torus$ is  
\be
J(\torus)=\int_{\torus} \langle j,\xi\rangle\, dV
\ee
where $\langle j, \xi\rangle=g(j,\xi)$ and $dV$ is the volume element of $g$.

As in Section \ref{MAINAPPS} define $R=R({\rm T})$ as the length of the greatest axisymmetric orbit projecting into $\Omega$, define $P=P({\rm T})$ to be the sectional perimeter of ${\rm T}$ and let $D=D({\rm T})$ be the distance from ${\rm T}$ to the axis. It is easily checked\footnote{Use that every $\gamma$-geodesic in $\Sigma$ can be (isometrically) lifted to a $g$-geodesic in $\Sigma$, and that the $g$-length of any curve in $\Sigma$ is greater or equal than the $\gamma$-length of its projection into $\Sigma$.} that $D$ is equal to the $\gamma$-distance inside $\Sigma$ from $\Pi({\rm T})$ to the axis $\partial \Sigma$. However, $P$ is not necessarily equal to the $\gamma$-perimeter $\tilde{P}=\tilde{P}({\rm T})$ of $\partial^{e} \Pi({\rm T})$ in $\Sigma$. Instead we only have $P\geq \tilde{P}$.   

There are two main tools that we will use to prove Theorems \ref{AST1} and \ref{ASPP}. The first is the following inequality,
\be\label{MAINE}
8\pi |J(\torus)|\leq \frac{R^{2}}{2\pi}\int_{\Sigma}\big[ |\nabla \alpha|^{2}+\kappa \alpha^{2}\big]\, dA
\ee
valid for any $\Omega$ and any $\alpha$ of compact support in $\Sigma^{
}$ with $\alpha\geq 1$ over $\Omega$. To see this just compute
\begin{align}
\nonumber |J(\torus)| & = \big|\int_{\torus} \langle j,\xi\rangle\, dV\big|\leq 2\pi \int_{\Omega} |j|\lambda^{2}\, dA \leq \frac{R^{2}}{2\pi} \int_{\Omega} |j|\, dA \\  
& \leq \frac{R^{2}}{2\pi} \int_{\Omega} \rho \alpha^{2}\, dA \leq \frac{R^{2}}{2\pi} \int_{\Sigma} \rho \alpha^{2}\, dA \leq \frac{R^{2}}{16\pi}\int_{\Sigma} \big[ |\nabla \alpha|^{2} +\kappa \alpha^{2}\big]\, dA
\end{align}
where we used $|j|\leq \rho$ and that for any orbit $C$ we have ${\rm length}(C)=2\pi \lambda(C)$.

The second tool is a fundamental estimation of the integral $\int (|\nabla \alpha|^{2}+\kappa\alpha^{2})dA$ when the trial functions $\alpha$ 
are chosen conveniently as {\it radial} functions. Let us explain how these functions are defined and which estimations we obtain out of them. Let $\Omega\subset {\rm Int}(\Sigma)$ be a region which is topological a two-disc. Then, for any $L<D$ define the domain
\be
\Omega_{L}:=\big\{p\in (\Sigma\setminus \Omega): \dist_{\gamma}(p,\Omega)\leq L\big\}
\ee
Thus, $\Omega_{L}$ is the set of points in the complement of $\Omega$ and at a distance less or equal than $L$ from $\Omega$ itself.

Now, define $\alpha:\Omega_{L}:\rightarrow \mathbb{R}$ by
\be
\alpha_{L}(p)=1-\frac{\dist_{\gamma}(p,\Omega)}{L}
\ee
The main estimation is that with this particular $\alpha$ (i.e. $\alpha=\alpha_{L}$) we have
\be\label{FB}
\int_{\Omega_{L}}\big[|\nabla\alpha|^{2}+\kappa\alpha^{2}\big]\, dA\leq \frac{2\tilde{P}}{L}+\tilde{P}'-\frac{A}{L^{2}} 
\ee
where $A={\rm Area}(\Omega_{L})$ and $\tilde{P}'$ is the first variation of $\tilde{P}$ in outwards direction to $\Omega$.
This is proved in Theorem 1 of arXiv:1002.3274.

On the other hand, if we define $\alpha=1$ on $\Omega$ then by Gauss-Bonet we obtain
\be\label{FB2}
\int_{\Omega}\big[|\nabla\alpha|^{2}+\kappa\alpha^{2}\big]\, dA=2\pi-\tilde{P}'
\ee
where $\tilde{P}'$ is the first variation of $\tilde{P}$ in the outwards direction to $\Omega$.

Combining (\ref{FB}) and (\ref{FB2}) we deduce that for the $H^{1}$-function $\alpha:{\rm Int}(\Sigma)\rightarrow \mathbb{R}$ of compact support given by
\be\label{FF}
\alpha(p)=
\left\{
\begin{array}{lcl}
1 & {\rm if} & p\in \Omega,\\
\alpha_{L}(p) & {\rm if} & p\in \Omega_{L},\\
0 & {\rm if} & p\in \Sigma\setminus (\Omega\cup\Omega_{L})
\end{array}
\right.
\ee
we have
\be\label{FB3}
\int_{\Sigma}\big[|\nabla\alpha|^{2}+\kappa\alpha^{2}\big]\, dA\leq \frac{2\tilde{P}}{L}+2\pi-\frac{A}{L^{2}}\leq \frac{2\tilde{P}}{L}+2\pi\leq \frac{2 P}{L}+2\pi,
\ee
where the last inequality follows because $\tilde{P}\leq P$.

\vs
We can use now the two tools just described to prove Theorem \ref{THSP}.

\begin{proof}[\bf Proof of Theorem \ref{THSP}] Let $\Omega$ be the region enclosed by $\partial^{e}(\Pi(\body))$. Let $L=D(\Omega)$. Then use (\ref{MAINE}) with $\alpha$ defined by (\ref{FF}), and then use (\ref{FB3}) to deduce 
\be
8\pi\, |J(\body)|\leq \frac{R^{2}}{2\pi}\bigg(2\pi +\frac{2P}{D}\bigg)=\bigg(1+\frac{P}{\pi D}\bigg)R^{2}
\ee
as wished. 
\end{proof}
 
To prove Theorem \ref{ASPP} we will use that with the trial function
\be\label{FF4}
\alpha(p)=
\left\{
\begin{array}{lcl}
1 - \frac{\displaystyle \dist_{\gamma}(p,\Pi(C))}{\displaystyle L} & {\rm if} & \dist_{g}(p,\Pi(C))<L,\vs\\
0 & {\rm if} & \dist_{\gamma}(p,\Pi(C))\geq L
\end{array}
\right.
\ee 
where $L<\dist_{g}(C,{\rm Axis})=\dist_{\gamma}(\Pi(C),\partial \Sigma)$, we have
\be\label{CINE}
\int_{\Sigma}\big[|\nabla\alpha|^{2}+\kappa\alpha^{2}\big]\, dA\leq 2\pi - \frac{A}{L^{2}}
\ee 
where we recall that $A=A(\Omega_{L})$. This inequality is obtained easily as a limit case of the inequality (\ref{FB3}) when $\Omega$ reduces to a point. It can also be obtained from Lemma 1.8 in Castillon's \cite{MR2225628}. Indeed, choosing $\xi(r)=1-r/L$ in Lemma 1.8 gives us the bound $\int \kappa\alpha^{2}dA\leq 2\pi-2A/L^{2}$, while because $|\nabla \alpha|=1/L$ we get $\int |\nabla\alpha|^{2}dA=A/L^{2}$.

\begin{proof}[\bf Proof of Theorem \ref{ASPP}] Take $\alpha$ equal to the function (\ref{FF4}) with $L=d$, times the constant $1/(1-a/d)$. 
As $\body\subset \torus(C,a)$ then $\alpha\geq 1$ on $\Pi(\body)$. We use then (\ref{MAINE}) together with the fact that with such $\alpha$ we have (use (\ref{FF4})),
\be
\int_{\Sigma}\big[|\nabla\alpha|^{2}+\kappa\alpha^{2}\big]\, dA\leq \frac{2\pi}{(1-a/d)^{2}}
\ee
to obtain 
\be
8\pi\, |J|\leq \bigg(\frac{1}{1-a/d}\bigg)^{2}R^{2}
\ee
as wished.
\end{proof}

\subsection{Proof of the related results.}

\begin{proof}[\bf Proof of Theorem \ref{ISESTT}.] First, from the definition of the Kommar angular momentum we have
\be
J(\body)=\frac{1}{8\pi} \int_{\ssurf} K(\zeta,\xi)\, dA
\ee
where $\zeta$ is a normal to $\ssurf$ inside $\Sigma$. Then, by Cauchy-Schwarz we obtain, (make $|J|=|J(\body)|$),
\be
|J|\leq \frac{\sqrt{A}}{8\pi}\bigg[\int_{\ssurf} |K|^{2}_{g} |\lambda|^{2}\, dA\bigg]^{1/2}
\ee 
But, $|\lambda|\leq R/(2\pi)$ and, by the energy constraint ${\rm R}_{g} = |K|_{g}$, ($\ssurf$ is in vacuum), where ${\rm R}_{g}$ is the scalar curvature of $g$.
Thus,
\be\label{PEQ}
|J|\leq \frac{\sqrt{A}\, R}{16\pi^{2}}\bigg[\int_{\ssurf} {\rm R}_{g}\, dA\bigg]^{1/2}
\ee 
Finally, as shown by Christodoulou and Yau in \cite{MR954405}, we have\footnote{There seems to be a factor of $2$ missing in the denominator of the r.h.s of equation (5) in \cite{MR954405}.} $\int_{S} {\rm R}_{g}\, dA\leq 24\pi$. Using this in (\ref{PEQ}) we get
\be
|J|\leq \frac{6^{1/2}}{8\pi^{3/2}}\, \sqrt{A}\, R
\ee
This is the first inequality in (\ref{ISEST}). To obtain the second inequality as well we need to prove that the area $A$ of $\ssurf$ is less or equal than $L^{2}$ where $L$ the distance from the north to the south poles of $\ssurf$. This is proved as follows.

For any $L_{1}\in (0,L)$ define $L_{2}=L-L_{1}$. Let $\mathcal{P}_{1}$ and $\mathcal{P}_{2}$ be the poles of $\ssurf$. Then, given $L_{1}$ define a function $\alpha$ by
\be\label{ALPHS}
\alpha(p)=\left\{
\begin{array}{lcl}
1-\frac{\displaystyle \dist_{\gamma}(p,\mathcal{P}_{1})}{\displaystyle L_{1}} & {\rm if} & \dist_{\gamma}(p,\mathcal{P}_{1})\leq L_{1},\vs \\
-1+\frac{\displaystyle \dist_{\gamma}(p,\mathcal{P}_{2})}{\displaystyle L_{2}} & {\rm if} & \dist_{\gamma}(p,\mathcal{P}_{2})\leq L_{2}
\end{array}
\right.
\ee
Note that the function $\alpha$ takes positive and negative values. It is clear too that for some $L_{1}$ in $(0,L)$ the integral of $\alpha$ on $\ssurf$ is zero. Denote such $L_{1}$ by $L^{*}_{1}$ and write $L_{2}^{*}=L-L_{1}^{*}$. The function (\ref{ALPHS}) for these values of $L_{1}$ and $L_{2}$ is denoted by $\alpha^{*}$. 

Now, the stability inequality for stable isoperimetric surfaces implies
\be
\int_{\ssurf} \big[ |\nabla \alpha^{*}|^{2}+\kappa\alpha^{*2}\big ]\, dA\geq 0
\ee
Using twice (\ref{CINE}), once for the integral on the domain $\{p: \dist_{\gamma}(p,\mathcal{P}_{1})\leq L^{*}_{1}\}$ and a second time for the domain $\{p: \dist_{\gamma}(p,\mathcal{P}_{2})\leq L^{*}_{2}\}$, we can bound the integral on the l.h.s of the previous equation by $(2\pi -A^{*}_{1}/L_{1}^{*2})+(2\pi-A_{2}^{*}/L^{*2}_{2})$, where $A^{*}_{i}$, $i=1,2$, are the areas of the domains $\{p: \dist_{\gamma}(p,\mathcal{P}_{i})\leq L^{*}_{i}\}$, $i=1,2$. Hence,
\be
4\pi -\frac{A^{*}_{1}}{L_{1}^{*2}}-\frac{A_{2}^{*}}{L_{2}^{*2}}\geq 0
\ee
and therefore,
\be
4\pi \geq \frac{A^{*}_{1}}{L_{1}^{*2}}+\frac{A_{2}^{*}}{L_{2}^{*2}}\geq \frac{A^{*}_{1}+A^{*}_{2}}{\max\{L^{*2}_{1},L_{2}^{*2}\}}\geq \frac{A}{L^{2}}
\ee
as wished.\end{proof}

\begin{proof}[\bf Proof of Theorem \ref{QATHM}.] For the proof we will use the following property of the Hawking energy on spherically symmetric spacetimes, (see for instance \cite{Bray:2006pz}). 

Let $({\bf M}; {\bf g}; {\bf T})$ be a spherically symmetric spacetime where it is assumed that ${\bf T}$ satisfies the dominant energy condition. Let $\varphi:\mathbb{S}^{2}\times [0,1]\rightarrow {\bf M}$ be a spacelike embedding for which every $\varphi(\mathbb{S}\times \{x\})$ is a rotationally invariant sphere. Define $\zeta=d\varphi(\partial_{x})/|d\varphi(\partial_{x})|$, ($x$ is the coordinate on $[0,1]$), and consider a unit-timelike vector $n$ normal to the image of $\varphi$ in ${\bf M}$. In this setup we have the following: If over every sphere $\varphi(\mathbb{S}^{2}\times \{x\})$, the null expansions $\theta^{+}$ and $\theta^{-}$ along the null vectors $l^{+}=\zeta+n$ and $l^{-}=\zeta-n$ respectively, are positive, then the Hawking energy at $\varphi(\mathbb{S}\times \{1\})$ is greater or equal than the Hawking energy at $\varphi(\mathbb{S}\times \{0\})$. Recall that the Hawking energy $H(\ssurf)$ of a rotationally symmetric  sphere $S$ is
\be
H(\ssurf)=\frac{\sqrt{A}}{16\pi}\bigg(1-\frac{1}{16\pi}\int_{\ssurf} \theta^{+}\theta^{-} dA\bigg)=\frac{\sqrt{A}}{16\pi}\bigg(1-\frac{\theta^{+}\theta^{-} A}{16\pi}\bigg)
\ee

We proceed with the proof of the Theorem \ref{QATHM}. Suppose first that $0<M<|Q|$, \footnote{By the positive energy theorem we always have $M>0$.}. Then, because of the spherical symmetry, the exterior of $\body$ in $\hyper$ is modelled as a slice of the Reissner-Nordstr\"om superextreme spacetime which, recall, has the metric
\be\label{RNM}
{\bf g}_{\rm RN}=-\bigg(1-\frac{2M}{r}+\frac{Q^{2}}{r^{2}}\bigg)dt^{2} +\frac{1}{\displaystyle \bigg(1-\frac{2M}{r}+\frac{Q^{2}}{r^{2}}\bigg)}dr^{2} +r^{2}d\Omega^{2}
\ee
on the range of coordinates $r\in (0,\infty)$, $t\in (-\infty,\infty)$, $\theta\in [0,\pi]$ and $\varphi\in (0,2\pi]$.
A simple computation then shows that $\theta^{+}$ and $\theta^{-}$ at $\partial \body$ are both positive, (this will be crucial below), and that, the Hawking energy $H(\partial \body)$ at $\partial \body$ is
\be\label{HWM}
H(\partial \body)=M-\frac{\sqrt{\pi} Q^{2}}{\sqrt{A(\partial \body)}}.
\ee

Now, if $\theta^{+}\neq 0$ and $\theta^{-}\neq 0$ at every rotationally invariant sphere in $\body$, then $\theta^{+}>0$ and $\theta^{-}>0$ at each one of them. Hence, we can use the property explained above to conclude that the Hawking energy at $\partial \body$ must be greater or equal than the Hawking energy at the origin of $\body$ which is zero, (think it as a degenerate sphere). Thus, $H(\partial \body)\geq 0$ in (\ref{HWM}), and (\ref{MQA}) then follows. 

If instead there is a rotationally invariant sphere in $\body$ having either $\theta^{+}=0$ or $\theta^{-}=0$, then, again by the same property explained above, the Hawking energy at $\partial \body$ must be greater or equal than the Hawking energy of the rotationally symmetric sphere in $\body$ which is closest to $\partial \body$, and which has either $\theta^{+}=0$ or $\theta^{-}=0$, \footnote{Again note that at each rotationally symmetric sphere between this last one and $\partial \body$ we have $\theta^{+}>0$ and $\theta^{-}>0$.}.  But the Hawking energy of this last sphere is positive because one of its null expansions is zero. Therefore $H(\partial\body)>0$, and (\ref{MQA}) follows also in this case. 

Let us assume now that $M\geq |Q|$. Let $r_{\partial \body}=\sqrt{A(\partial \body)/4\pi}$ be the areal-coordinate at $\partial \body$. If 
\be\label{APEN}
r_{\partial \body}\geq M-\sqrt{M^{2}-Q^{2}}
\ee
then we are done because $M-\sqrt{M^{2}-Q^{2}}>Q^{2}/2M$ which together with (\ref{APEN}) implies (\ref{MQA}). If not, then $r_{\partial \body}< M-\sqrt{M^{2}-Q^{2}}$, that is, the areal-coordinate $r_{\partial \body}$ is less than the smaller root of the polynomial $r^{2}-2Mr+Q^{2}$. For this reason, a small neighbourhood of $\partial \body$ in $\hyper\setminus {\rm Int}(\body)$, (that is, in the exterior of the body), can be modelled as a slice of the piece of the Reissner-Nordstr\"om spacetime given by the metric (\ref{RNM}) in the range of coordinates $r\in (0,M-\sqrt{M^{2}-Q^{2}})$, $t\in (-\infty,\infty)$, $\theta\in [0,\pi]$ and $\varphi\in (0,2\pi]$. But then the null expansions $\theta^{+}$ and $\theta^{-}$ at $\partial \body$ must be again positive\footnote{If they are both negative, (which is the only other option), then the slice $\hyper$ outside $\body$ must reach the singularity at $r=0$.} and we can repeat exactly the same argument as we did for the case $M<|Q|$. \end{proof}

\section{Acknowledgment.} I would like to thank Sergio Dain for important conversations and to the many colleagues of FaMAF (Argentina) where these results were first discussed.  

\bibliographystyle{plain}
\bibliography{Master}

\end{document}